\documentclass[12pt]{article}

\usepackage{amsfonts}
\usepackage{amsmath}
\usepackage{amssymb}
\usepackage{color}
\usepackage{graphicx}
\usepackage{float}
\usepackage[mathscr]{eucal}
\usepackage{amsfonts}
\usepackage{float}
\usepackage[colorlinks,linkcolor=blue]{hyperref}
\usepackage[latin1]{inputenc}
\usepackage{amsmath}
\usepackage{graphicx}
\usepackage{amssymb}
\usepackage{amsthm}
\usepackage{verbatim}
\usepackage{epsfig}
\usepackage[labelfont=bf]{caption}
\usepackage{enumerate}
\usepackage{subfig}
\usepackage{epstopdf}
\usepackage{algorithmicx,algorithm}

\newtheorem{lemma}{Lemma}
\newtheorem{proposition}{Proposition}
\newtheorem{corollary}{Corollary}
\newtheorem{fact}{Fact}

\newtheorem{remark}{Remark}

\newtheorem{assumption}{Assumption}

\def\begcen{\begin{center}}
\def\endcen{\end{center}}

\newcommand{\bfq}{\mbox{$q$}}

\newcommand{\bfy}{\mbox{$y$}}

\newcommand{\col}{ \mbox{col} }

\def\caly{{\cal Y}}
\def\calg{{\cal G}}

\def\bfy{{\bf y}}

\def\tilthe{\tilde{\eta}}

\def\liminf{\lim_{t \to \infty}}

\def\L2{{\cal L}_2}
\def\L2e{{\cal L}_{2e}}

\def\rea{\mathbb{R}}

\def\adj{\mbox{adj}}

\def\tilthe{\tilde{\theta}}

\def\x{{x}}

\def\beal#1{\begin{align}{#1}\end{align}}
\def\begmat#1{\begin{bmatrix}#1\end{bmatrix}}
\def\qmx#1{\begin{bmatrix}#1\end{bmatrix}}
\def\begali#1{\begin{align}{#1}\end{align}}
\def\begalis#1{\begin{align*}{#1}\end{align*}}

\def\begequarr{\begin{eqnarray}}
\def\endequarr{\end{eqnarray}}
\def\begequarrs{\begin{eqnarray*}}
\def\endequarrs{\end{eqnarray*}}
\def\begarr{\begin{array}}
\def\endarr{\end{array}}
\def\begequ{\begin{equation}}
\def\endequ{\end{equation}}
\def\lab{\label}
\def\begdes{\begin{description}}
\def\enddes{\end{description}}
\def\begenu{\begin{enumerate}}
\def\begite{\begin{itemize}}
\def\endite{\end{itemize}}
\def\endenu{\end{enumerate}}

\def\lef[{\left[\begin{array}}
\def\rig]{\end{array}\right]}

\def\begcen{\begin{center}}
\def\endcen{\end{center}}
\def\begrem{\begin{remark}\rm}
\def\endrem{\end{remark}}
\def\begassum{\begin{assumption}}
\def\endassum{\end{assumption}}
\def\begassums{\begin{assumption*}}
\def\endassums{\end{assumption*}}
\def\begassu{\begin{ass}}
\def\endassu{\end{ass}}
\def\beglem{\begin{lemma}}
\def\endlem{\end{lemma}}
\def\begcor{\begin{corollary}}
\def\endcor{\end{corollary}}
\def\begfac{\begin{fact}}
\def\endfac{\end{fact}}
\def\TAC{{\it IEEE Trans. Automat. Contr.}}
\def\AUT{{\it Automatica}}

\def\SCL{{\it Systems and Control Letters}}

\def\liminf{\lim_{t \to \infty}}

\def\L2e{{\cal L}_{2e}}

\def\rea{\mathbb{R}}
\def\intnum{\mathbb{Z}}

\def\adj{\mbox{adj}}
\def\col{\mbox{col}}

\def\JPC{{\it Journal of Process Control}}

\def\TAC{{\it IEEE Trans. Automatic Control}}

\def\IJC{{\it International Journal of Control}}

\def\SCL{{\it Systems \and Control Letters}}
\def\AUT{{\it Automatica}}

\def\SIAM{{\it SIAM J. Control and Optimization}}

\def\IETGTD{{\it  IET Generation, Transmission \and Distribution}}

\def\ba{\begin{array}}
\def\ea{\end{array}}

\def\begsubequ{\begin{subequations}}
\def\endsubequ{\end{subequations}}

\def\bfthe{{\boldsymbol\theta}}
\def\bfq{{\bf Q}}
\def\bfe{{\bf e}}
\def\calw{{\cal W}}
\def\calk{{\cal K}}
\def\bfome{{\boldsymbol{\Omega}}}
\def\bfeps{{\boldsymbol \epsilon}}
\def\bfy{{\bf Y}}
\def\tilthe{\tilde{\bfthe}}

\graphicspath{{figs/}}

\title{An Adaptive Observer for Uncertain Linear Time-Varying Systems with Unknown Additive Perturbations}

\author{Anton Pyrkin, Alexey Bobtsov, Romeo Ortega and Alberto Isidori}

\begin{document}

\maketitle


\begin{abstract}                         
In this paper we are interested in the problem of {\em adaptive state observation} of linear time-varying (LTV) systems where the system and the input matrices depend on unknown time-varying parameters.  It is assumed that these parameters satisfy some known LTV dynamics, but with unknown initial conditions.  Moreover, the state equation is perturbed by an additive signal generated from an exosystem with uncertain constant parameters.  Our main contribution is to propose a {\em globally convergent} state observer that requires only a weak excitation assumption on  the system. 
\end{abstract}

%
\section{Introduction}
\lab{sec1}
%
In view of its wide practical application the  problem of disturbance cancellation, has been extensively studied in the literature. The interested reader is referred to \cite{YILBAS} for a recent, comprehensive review of the literature. A scenario that is now universally adopted to tackle this problem is the {\em output regulation} paradigm---formally articulated in \cite[Chapter 4]{ISIbook}---which treats in a unified framework the problems of output tracking and disturbance rejection. In this scenario it it is assumed that the disturbance to be rejected (or the output to be tracked) is generated by an exosystem that is an autonomous dynamical system. It is assumed that only part of the state is available for measurement, which constitutes the output signal. Consequently, a first task to be solved in the output regulation problem is the one of state observation.

In this paper we adopt the mathematical formulation of output regulation addressing the problem of {\em adaptive state observation} of LTV systems where the system and the input matrices depend on unknown time-varying parameters and the state equation is perturbed by an {\em additive signal} generated from an exosystem with uncertain constant parameters.\footnote{See  \cite{DOSetalbook} for an extensive survey on the problem of state observation of LTV systems.} More precisely, in the paper we assume the following conditions.
\begenu[{\bf C1}]
\item The system is LTV with unknown time-varying parameters.
\item These parameters satisfy an autonomous LTV system with known dynamics but unknown initial conditions.
\item The measurable output is only one of the components of the state vector.
\item The disturbance to be rejected is generated by an exosystem whose system matrix depends on a vector of unknown constant parameters.  
\item The objective is to design an adaptive observer that will reconstruct the system state and all the unknown parameters.  
\endenu 

The problem formulation adopted in the paper is very similar to the estimation part of the work \cite{YILBAS}.\footnote{It should be clarified that in  \cite{YILBAS}, besides the estimation of the disturbance, a step of control via adaptive backstepping is added, while in this paper we restrict ourselves to the observation part.} However, it is important to underscore that some critical assumptions \cite[Assumptions 1-4]{YILBAS} are avoided in the present paper. Specifically, in  \cite{YILBAS} it is assumed that
\begite
\item the plant is minimum-phase;
\item the sign of the high-frequency gain is known;
\item the relative degree of the plant and upperbound of the plant order are known;
\item the disturbance is the sum of a known amount of harmonics.
\endite
In the present work we are able to work with nonminimum-phase system. The high-frequency gain may be unknown. We assume only that the plant order is given, relative degree may be uncertain. Disturbance in the present work is arbitrary for a class of LTI exogenous systems.
The present work constitutes a non-trivial extension of our previous research, e.g. \cite{PYRTDS,PYRTAC,PYRARC}, where we addressed the problems of disturbance estimation and cancellation for a plant with input delay and partial measurement of the state. But in mentioned works we considered linear time invariant systems with known constant parameters and all uncertaities were localized in the disturbance. In this paper we consider the system with unknown time-varying parameters and uncertain disturbance.

To tackle the problem at hand we propose in this paper to use a new procedure to design state observers for state-affine systems, called parameter estimation-based observers (PEBO), first reported in \cite{ORTetalscl}.  The main novelty of PEBO is that the state observation problem is reformulated as a problem of {\em parameter estimation} of a regressor equation. A drawback of the PEBO proposed in  \cite{ORTetalscl} is that it involves a non-robust open-loop integration. This shortcoming was later removed in \cite{ORTetalaut} with the definition of the generalized (G)PEBO, where the properties of the {\em principal matrix solution}  of an unforced LTV system $\dot x(t) = A(t)x(t)$ are exploited. This novel technique has been successfully applied to reaction systems \cite{ORTetaljpc}, power systems \cite{LORetal}, systems with delayed measurements \cite{BOBetal} and distributed state estimation \cite{ORTNUNBOB}---see also \cite{PYRetalscl} for a related adaptive observer design for a class of nonlinear systems.  

In this paper we use the GPEBO technique to derive the regressor equation needed for the estimation of the state and the system parameters---including the parameters of the aforementioned exosystem. To estimate these parameters we use the dynamic regressor extension and mixing (DREM) parameter adaptation algorithm introduced in \cite{ARAetaltac} and later further elaborated in \cite{ORTetaltac}.\footnote{The interested readers are referred to the aforementioned papers for further details on GPEBO and DREM.} One central feature of DREM is that it generates, out of a $q$-dimensional regression equation, $q$ {\em scalar} regression equations. This property is essential in our problem since, as we will show below, the regressor equation associated to this problem is nonlinear and it consists of a linear and a nonlinear part, and we need to invoke DREM to concentrate our attention on the linear part of the regression equation. 

Combining GPEBO and DREM has already been explored by the authors using gradient estimators in \cite{BOBetalaut,KORetalscl}. An important contribution of the present work is that we propose to  use---instead of the gradient scheme---a {\em new least-squares (LS)-based} algorithm, referred in sequel as  [LS+DREM] estimator. The use of LS advocated in the paper removes the need to calculate a, computationally demanding state-transition matrix in the estimator. Moreover, the superior convergence properties of LS estimators, as opposed to gradient-based, are widely recognized \cite{LJUbook,RAOTOUbook}. In the paper  we show that the parameters of the regressor equation can be---globally and exponentially---estimated under the weak assumption of {\em interval excitation} (IE) \cite{KRERIE,TAObook} of the regressor. This should be contrasted with the highly demanding persistency of excitation required in classical (non-DREM-based) gradient and LS estimators  \cite{IOASUNbook,SASBODbook}. 

The remainder of the paper is organized as follows. In Section \ref{sec2} we present the problem formulation.  In Section \ref{sec3} we derive the regression equation used in the estimator, which is presented in Appendix B. The main result is given in Section   \ref{sec4}. In Section \ref{sec5} we summarize all the equations needed for the implementation of the adaptive observer. We wrap up the paper with concluding remarks in Section   \ref{sec6}. In Appendix A we give some preliminary Lemmata instrumental for the proof of our main results.\\

\noindent {\bf Notation.}  Given $n \in \intnum_{+},q \in \intnum_{+}$, $I_n$ is the $n \times n$ identity matrix and ${ 0}_{n\times q}$  is an $n \times q$ matrix of zeros. For $x \in \rea^n$, we denote the square of the Euclidean norm as $|x|^2:=x^\top x$.  $\bfe_q \in \rea^n$ denotes the $q$-th vector of the $n$-dimensional Euclidean basis. Given a smooth signal $z:\rea_+ \to \rea$ we denote $[z(t)]^{(i)}=:{d^i z(t)\over dt^i}$, for $i \in \intnum_{+}$. 
%
\section{Problem Formulation}
\lab{sec2}
%
In this paper we are interested in the problem of {\em adaptive state observation} of  the following uncertain, additively perturbed,  LTV system 
\begali{
	\nonumber
	\dot x(t) &=[A + \theta(t)\bfe^\top_1] x(t)+B(t)u(t)+\bfe_n\delta(t),\\
	\lab{plant}
		y(t) &= \bfe^\top_1 x(t),
}
with state $x(t)\in\rea^n$, input $u(t)\in\rea$ and measurable output $y(t)\in\rea$, where 
$$
A:=\begmat{0_{(n-1)\times 1} & & I_{n-1}\\&0_{1 \times n} & }.
$$
It is assumed that  the unknown time-varying vectors $\theta(t)\in\rea^n$ and $B(t)\in\rea^n$ satisfy the LTV dynamics
\begali{
\nonumber
\dot x_\theta(t)&=A_\theta(t)x_\theta(t), \\
\nonumber
\dot x_B(t)&=A_B(t)x_B(t),\\
\nonumber
\theta(t)&=h_\theta x_\theta(t), \\
\lab{ltvsys}
B(t)&=h_Bx_B(t), 
}
with $x_\theta(t) \in\rea^{n_\theta}$, $x_B(t) \in\rea^{n_B}$, and the matrices $h_\theta \in\rea^{n \times n_\theta}$, $h_B \in\rea^{n \times n_B}$, $A_\theta(t) \in\rea^{n_\theta \times n_\theta}$, $A_B(t) \in\rea^{n_B \times n_B}$ are {\em known}. On the other hand,   the disturbance $\delta(t) \in \rea$ is the output of the exogenous system
\beal{
\nonumber
\dot w(t) &= S(\rho) w(t)\\
	\lab{delta}
\delta(t) &=h_\delta^\top w(t),
}
with $w(t)\in\rea^{n_w}$. It is assumed that the mapping $S:\rea^{n_\rho} \to \rea^{n_w \times n_w}$ is known, but the vectors of constant parameters $\rho\in\rea^{n_\rho}$ and $h_\delta\in\rea^{n_w}$ are {\em unknown}. Furthermore, without loss of generality \cite{ISIbook}, it is assumed that all the eigenvalues of $S(\rho)$ have {\em nonnegative} real part.

The objective is to reconstruct asymptotically (or in finite time) the vectors $\theta(t)$, $B(t)$, $\rho$, and the state $x(t)$.

\begrem
\lab{rem1}
It should be underscored that the initial conditions of all the differential equations that describe the systems dynamics are {\em unknown}.
\endrem 
%
\section{Reducing the Observation Problem to Estimation of Parameters}
\lab{sec3}
%
In this section we apply the GPEBO procedure \cite{ORTetalaut} to translate the problem of state observation to one of parameter estimation that will be combined with the task of estimation of the system parameters, to develop an adaptive observer.  As will be shown below, the regression equation that should be used for parameter estimation is {\em nonlinear} in the parameters. Indeed, it consists of the sum of a ``classical" linear regression equation (LRE) with {\em all} the unknown constant parameters and a term which depends nolinearly on these parameters. As it has been shown in \cite{ORTetalaut21}, thanks to the use of the DREM technique---that generates {\em scalar} LREs for separable nonlinearly parameterized regressions---it is possible to use for the parameter estimation only the first, linear part of  the regression equation and disregard the nonlinear part of it.    

\begin{proposition}
\lab{pro1}\em
There exists {\em measurable} signals $z(t)\in\rea^n$, $\bfome_L(t)\in\rea^s$, $\bfome_N(t)\in\rea^{n \cdot s}$, $\bfy(t)\in\rea$,  two known mappings  $\calw: \rea_+ \times \rea^s\to\rea^n$,   $\calk: \rea^s\to\rea^{n \cdot s}$ and a vector of {\em constant} parameters $\bfthe\in\rea^s$ such that the state of the system  \eqref{plant} can be written as
\begequ
\lab{xt}
x(t)=z(t)+\calw(t,\bfthe)+\bfeps_x(t)
\endequ
where $\bfthe$ satisfies the {\em separable}, nonlinearly parameterized regression equation
\begequ
\lab{lre}
\bfy(t)=\bfome_L^\top (t)\bfthe + \bfome_N^\top(t) \calk(\bfthe)  + \varepsilon(t)
\endequ
with $\bfeps_x(t)\in\rea^n$ and $\varepsilon(t)\in\rea$ exponentially decaying to zero and $s:=n+n_\theta+n_B$.
\end{proposition}

\begin{proof}
Following the GPEBO construction we invoke \cite[Property 4.4]{RUGbook} to rewrite the vectors ${x_\theta}(t)$ and ${x_B}(t)$ of \eqref{ltvsys} as
\begalis{
{x_\theta}(t)&=\Phi_\theta(t){x_\theta}_0,\\
{x_B}(t)&=\Phi_B(t){x_B}_0,
}
${x_\theta}_0={x_\theta}(0)$ and ${x_B}_0={x_B}(0)$, and $\Phi_\theta(t)$, $\Phi_B(t)$ are the principal matrix solutions of the LTV systems \eqref{ltvsys}, which are the solutions of the auxiliary filters
\begali{
\nonumber
\dot\Phi_\theta(t)&=A_\theta(t)\Phi_\theta(t), \quad
\Phi_\theta(0)=I_{n_\theta},\\
\lab{con1}
\dot\Phi_B(t)&=A_B(t)\Phi_B(t), \quad
\Phi_B(0)=I_{n_B}.
}
Thus, the vectors $\theta(t)$ and $B(t)$ may be rewritten as
\begalis{
\theta(t) & =h_\theta \Phi_\theta(t) {x_\theta}_0 \\
B(t) &=h_B\Phi_B(t){x_B}_0.
}
Choose a  vector $K\in \rea^n$  such that the matrix   
\begequ
\lab{ak}
A_K:=\qmx{ & I_{n-1} \\
	-K & \\
	& 0_{1 \times n-1}}.
\endequ
is Hurwitz, and define the set of filters
\begali{
\nonumber
	\dot z(t) & = A_Kz(t)+Ky(t),\\
\nonumber
\dot\Omega(t) &=A_K\Omega(t)+h_\theta \Phi_\theta(t) y(t),\\
\lab{con2}
\dot P(t) &= A_K P(t)+h_B\Phi_B(t)u(t),
}
with states $z(t)\in\rea^{n}$, $\Omega(t)\in\rea^{n\times {n_\theta}}$ and $P(t)\in\rea^{n\times{n_B}}$.

Define the auxiliary variable
\beal{
\lab{e}
e(t):=x(t)-z(t)-\Omega(t)\,{x_\theta}_0-P(t)\,{x_B}_0
}
whose derivative is given by\footnote{To simplify the notation, whenever clear from the context, the time time argument will be omitted.}
\begalis{
\dot e(t) &=Ax+\theta y+Bu+\bfe_n\delta-[A_Kz+Ky]-[A_K\Omega +h_\theta\Phi_\theta y]{x_\theta}_0 -[A_KP+h_B\Phi_Bu]{x_B}_0,\\
&=Ax+h_\theta \Phi_\theta(t) {x_\theta}_0  y+h_B\Phi_B(t){x_B}_0 u+\bfe_n\delta-[A_Kz+Ky]-[A_K\Omega +h_\theta\Phi_\theta y]{x_\theta}_0 \nonumber\\
&\quad-[A_KP+h_B\Phi_Bu]{x_B}_0,\\
&=Ax +\bfe_n\delta-[A_Kz-K\bfe^\top_1  x ]-A_K\Omega {x_\theta}_0-A_KP{x_B}_0,\\
&=A_K ( x-z-\Omega\,{x_\theta}_0-P\,{x_B}_0)+\bfe_n\delta,
}
consequently
\begequ
\lab{dote}
\dot e(t) =A_K e(t) +\bfe_n h^\top_\delta w(t).
\endequ
Combining the equation above with \eqref{delta} we get
\begalis{
\dot w(t) &= S(\rho)w(t)\\
\dot e(t) &= A_Ke(t) +\bfe_n h^\top_\delta w(t).
}
We apply now to this cascaded system the transformation discussed in Lemma 1 of Appendix A. Towards this end, consider the Sylvester equation
\[
\Pi(\rho) S(\rho) =A_K \Pi(\rho) +\bfe_n h^\top_\delta,
\]
that, given the fact that the spectra of $S(\rho)$ and $A_K$ are disjoint, has a unique solution. We underscore the fact that that $\Pi (\rho)$ depends on the unknown vector $\rho$, hence it is unknown. Defining the signal 
$$
\epsilon(t):= e(t)-\Pi(\rho) w(t),
$$
we get the decoupled dynamics
\begali{
\nonumber
\dot w(t)  &= S(\rho)w(t) \\
\dot {\epsilon}(t)  &= A_K\epsilon(t). 
\lab{dotwdote}
}
Hence
\[
e(t) = \Pi (\rho)w(t) + \epsilon(t)
\]
in which the second term is exponentially decaying. Multiplying the equation above by $\bfe^\top_1$ we get
\begalis{
\bfe^\top_1 e(t)&=\bfe^\top_1  [ \Pi (\rho)w(t) + \epsilon(t)]\\
& =\bfe^\top_1  [x(t) -  z(t) - \Omega(t){x_\theta}_0- P(t){x_B}_0 ],
}
where we have used \eqref{e} to get the second identity. Defining the {\em measurable} signals
\begali{
\nonumber
\zeta(t) &:= y(t) - \bfe^\top_1  z(t), \\
\lab{con3}
\varphi(t) &:={ \begmat{ \Omega^\top(t) \bfe_1   \\  P^\top(t) \bfe_1  }},
}
we obtain the first key relationship
\beal{
	\lab{y_delta}
\zeta(t)=\varphi^\top(t)\begmat{{x_\theta}_0 \\{x_B}_0}+\bfe^\top_1   \Pi (\rho)w(t) + \bfe^\top_1 \epsilon(t).
}
{
In the next step we will reparameterize \eqref{y_delta} in a different way to represent the unknown term $ \Pi (\rho)w(t)$ in a {\em bona fide} regression equation form.  Towards this end, consider the auxiliary (not-realizable) filter with the state $\Psi(t)  \in  \rea^{n_w}$:
\beal{
	\lab{Psi}
\dot\Psi(t)=A_f\Psi(t)+\bfe_n  \left[\bfe^\top_1  \Pi (\rho) w(t) + \bfe^\top_1 \epsilon(t)\right],
}
where $A_f \in  \rea^{n_w \times n_w}$ is given as
\begequ
\lab{af}
A_f:=\begmat{0_{(n-1)\times 1} & & I_{n-1}\\& f^\top & },\qquad 
\endequ
for some constant vector $f \in \rea^n$ such that $A_f$ is a Hurwitz matrix.  Using \eqref{y_delta} the system \eqref{Psi} may be rewritten as
\beal{
	\lab{Psi1}
	\dot\Psi(t)=A_f\Psi(t)+\bfe_n \left[\zeta(t)-\varphi^\top(t)\begmat{{x_\theta}_0 \\{x_B}_0}\right].
}
On the other hand,  invoking Lemma 2 of Appendix A, the system \eqref{Psi} may be rewritten as
\beal{
	\lab{Psi2}
	\dot\Psi(t)=A_\Gamma \Psi(t)+\bfe_n \varepsilon(t).
}
where $\varepsilon(t)$ is exponentially decaying and $A_\Gamma$ is given as
\begequ
\lab{agam}
A_\Gamma:=\begmat{0_{(n-1)\times 1} & & I_{n-1}\\& \Gamma^\top & },
\endequ
for some constant vector $\Gamma \in \rea^n$ such that   the spectra of $A_\Gamma $ and $S$ {\em coincide}. Comparing \eqref{Psi1} and \eqref{Psi2}, and taking into account the special form of matrices the $A_f$ and $A_\Gamma $, one can deduce that the last rows in the right hand sides of \eqref{Psi1} and \eqref{Psi2} are equal, that is,
\beal{
\lab{LRE_0}
f^\top\Psi(t)+\zeta(t)-\varphi^\top(t)\begmat{{x_\theta}_0 \\{x_B}_0}=\Gamma ^\top\Psi(t)+\varepsilon(t).
}

To remove the unknown vector $\Gamma$ from the identity above we introduce the {\em realizable} filters with states $L\in\rea^{n_w}$ and $Q\in\rea^{n\times(n_\theta+n_B)}$, 
\begali{
\nonumber
	\dot L(t) &= A_f L(t)+\bfe_n \zeta(t),\\
\lab{con4}
	\dot Q(t) & = A_f Q(t)+\bfe_n \varphi^\top(t)
}
and define the signal
\beal{
	\lab{E}
	E(t):=\Psi(t)-L(t)+Q(t)\begmat{{x_\theta}_0 \\{x_B}_0},
}
which, in fact, is exponentially decaying since it satisfies 
\begalis{
	\dot E(t) &= A_f\Psi+\bfe_n \zeta-\bfe_n \varphi^\top\begmat{{x_\theta}_0 \\{x_B}_0}-A_f L-\bfe_n \zeta+A_fQ\begmat{{x_\theta}_0 \\{x_B}_0}+\bfe_n \varphi^\top\begmat{{x_\theta}_0 \\{x_B}_0}\\
	&=A_f\left[\Psi-L+Q\begmat{{x_\theta}_0 \\{x_B}_0}\right]\\
	&=A_f\, E(t).
}
From \eqref{E} we have
\beal{
	\lab{Psi_eq}
	\Psi(t)=L(t)-Q(t)\begmat{{x_\theta}_0 \\{x_B}_0}+E(t),
}
Substitution of \eqref{Psi_eq} into \eqref{LRE_0} yields the second key relationship
\beal{
	\lab{LRE_00}
	(f-\Gamma) ^\top\left[L(t)-Q(t)\begmat{{x_\theta}_0 \\{x_B}_0}+E(t)\right]+\zeta(t)-\varphi^\top(t)\begmat{{x_\theta}_0 \\{x_B}_0}=\varepsilon(t),
}
which one can rewrite in the nonlinear regression form \eqref{lre} with the definitions
\begali{
\nonumber
\bfy(t)&:=\zeta(t)+L^\top(t) f, \\
 \bfome_L(t)&:= \begmat{ Q^\top(t) f+\varphi(t)\\ L(t)},
\lab{con5}
}
with the unknown parameter vector
\begequ
\lab{bfthe}
 \bfthe:= \begmat{{x_\theta}_0 \\{x_B}_0\\ \Gamma },
\endequ
for the linear term and
$$
\bfome_N^\top(t) \calk(\bfthe):=-\Gamma^\top Q(t)\begmat{{x_\theta}_0 \\{x_B}_0}.
$$
for the quadratic one, which can be clearly written as a separable NLPR via
\begequ
\lab{sepnlpre}
\Gamma^\top Q(t)\begmat{{x_\theta}_0 \\{x_B}_0}=\sum_{i=1}^n\sum_{j=1}^{n_\theta+n_B}Q_{ij}(t)\bfthe_i\bfthe_j.
\endequ

To complete the proof we need prove the relation \eqref{xt}. Towards this end, we recall the dynamical model of $e(t)$ \eqref{dote}	and the key relation \eqref{y_delta}, that we repeat here for ease of reference as
\begalis{
\dot e(t)&=A_Ke(t)+\bfe_n h_\delta w(t),\\
\bfe^\top_1  e(t)&=\zeta(t)-\varphi^\top(t)\begmat{{x_\theta}_0 \\ {x_B}_0},
}
where $A_K$ is given in \eqref{ak}. Then, using the fact that $\bfe^\top_1  A_K^i \bfe_n\equiv 0$ for all $i=0, ..,n-2$, we can compute
\begalis{
	e_1(t) &= \bfe^\top_1  e(t), \\
	\dot e_1(t) & = \bfe^\top_1  \dot e(t) = \bfe^\top_1  A_K e(t) +\underbrace{\bfe^\top_1  \bfe_n}_{=0} h^\top_\delta w(t) = \bfe^\top_1  A_K e(t), \\
	\dot e_2(t) & = \bfe^\top_1  \ddot e(t) = \bfe^\top_1  A_K^2 e +\underbrace{\bfe^\top_1  A_K\bfe_n }_{=0}h^\top_\delta w(t)  = \bfe^\top_1  A_K^2 e(t), \\
	& \quad \vdots \nonumber\\
	e_1^{(n-1)}(t) & = \bfe^\top_1  e^{(n-1)}(t) = \bfe^\top_1  A_K^{n-1} e(t) +\underbrace{\bfe^\top_1  A_K^{n-2}\bfe_n }_{=0}h^\top_\delta w(t)  = \bfe^\top_1  A_K^{n-1} e(t), 
}
that we write compactly as
\begequ
\lab{e1}
\qmx{\bfe^\top_1   e (t)\\ \bfe^\top_1  \dot e(t) \\ \vdots \\ \bfe^\top_1  e^{(n-1)}(t)}=\qmx{\bfe^\top_1  \\ \bfe^\top_1  A_K \\ \vdots \\ \bfe^\top_1  A_K^{n-1}} e(t)
\endequ

In the sequel we will express the left hand side vector of \eqref{e1} in an alternative form.  For, we rearrange  \eqref{LRE_0} as
$$
\zeta(t)-\varphi^\top(t)\begmat{{x_\theta}_0 \\ {x_B}_0}=(\Gamma -f)^\top \Psi (t)+\varepsilon(t),
$$
which replaced in \eqref{y_delta} yields
\begequ
\lab{bfe1e}
\bfe^\top_1  e(t)=(\Gamma -f)^\top \Psi (t)+\varepsilon(t).
\endequ
Now, the dynamics of $\Psi$ and $\varepsilon$ are defined in  \eqref{Psi2} and \eqref{dotxi}, \eqref{vareps1}, respectively. For ease of reference we repeat them here as
\begalis{
	\dot\Psi(t)&=A_\Gamma \Psi(t)+\bfe_n \varepsilon(t),\\
		\dot\xi(t) &= F_c\,\xi(t). \\
		\varepsilon(t)&=h_\varepsilon \,\xi(t).
}
Using these equations and \eqref{bfe1e} we can compute the left hand side vector of \eqref{e1} as
\beal{
	\lab{edots}
	\qmx{\bfe^\top_1   e(t) \\ \bfe^\top_1  \dot e (t)\\ \vdots \\ \bfe^\top_1  e^{(n-1)}(t)}
	=\qmx{(\Gamma -f)^\top \Psi(t) +\varepsilon(t)\\ (\Gamma -f)^\top \dot\Psi(t)+\dot\varepsilon(t) \\ \vdots \\ (\Gamma -f)^\top \Psi^{(n-1)}(t)+\varepsilon^{(n-1)}(t)}
	=\qmx{(\Gamma -f)^\top  \\ (\Gamma -f)^\top A_\Gamma  \\ \vdots \\ (\Gamma -f)^\top A_\Gamma ^{n-1}}\Psi(t)+\bar\varepsilon(t)
}
where  $\bar\varepsilon(t)$ is a suitably defined exponentially decaying signal. Equating the right hand sides of \eqref{e1} and \eqref{edots} yields the identity
\beal{
	\lab{ePsi}
	\qmx{\bfe^\top_1  \\ \bfe^\top_1  A_K \\ \vdots \\ \bfe^\top_1  A_K^{n-1}} e(t)=\qmx{(\Gamma -f)^\top  \\ (\Gamma -f)^\top A_\Gamma  \\ \vdots \\ (\Gamma -f)^\top A_\Gamma ^{n-1}}\Psi(t)+\bar\varepsilon(t),
}
where we notice that the matrix premultiplying $e(t)$ is full rank. Using the latter fact, recalling   \eqref{e} and replacing it in  \eqref{ePsi} finally yields
\begalis{
x(t) &=\qmx{\bfe^\top_1  \\ \bfe^\top_1  A_K \\ \vdots \\ \bfe^\top_1  A_K^{n-1}}^{-1}\qmx{(\Gamma -f)^\top  \\ (\Gamma -f)^\top A_\Gamma  \\ \vdots \\ (\Gamma -f)^\top A_\Gamma ^{n-1}}\Psi(t)+\qmx{\bfe^\top_1  \\ \bfe^\top_1  A_K \\ \vdots \\ \bfe^\top_1  A_K^{n-1}}^{-1}\bar\varepsilon(t)+z(t)+\Omega(t)\,{x_\theta}_0+P(t)\,{x_B}_0\\
&=\qmx{\bfe^\top_1  \\ \bfe^\top_1  A_K \\ \vdots \\ \bfe^\top_1  A_K^{n-1}}^{-1}\qmx{(\Gamma -f)^\top  \\ (\Gamma -f)^\top A_\Gamma  \\ \vdots \\ (\Gamma -f)^\top A_\Gamma ^{n-1}}\left[L(t)-Q(t)\begmat{{x_\theta}_0 \\{x_B}_0}+E(t)\right]+\qmx{\bfe^\top_1  \\ \bfe^\top_1  A_K \\ \vdots \\ \bfe^\top_1  A_K^{n-1}}^{-1}\bar\varepsilon(t)\nonumber\\
&\quad+z(t)+\Omega(t)\,{x_\theta}_0+P(t)\,{x_B}_0,
}
where have used  \eqref{Psi_eq} to replace $\Psi(t)$ in the second identity.  The proof of the claim  \eqref{xt} is completed from the equation above defining
\begalis{
\calw(t,\bfthe)&:=\qmx{\bfe^\top_1  \\ \bfe^\top_1  A_K \\ \vdots \\ \bfe^\top_1  A_K^{n-1}}^{-1}\qmx{(\Gamma -f)^\top  \\ (\Gamma -f)^\top A_\Gamma  \\ \vdots \\ (\Gamma -f)^\top A_\Gamma ^{n-1}}\,\left\{L(t)-\left[Q(t)-\begmat{\Omega(t) & P(t)}\right]\qmx{x_{\theta_0}\\x_{B_0}}\right\},\\
\bfeps_x(t)&:=\qmx{\bfe^\top_1  \\ \bfe^\top_1  A_K \\ \vdots \\ \bfe^\top_1  A_K^{n-1}}^{-1}\qmx{(\Gamma -f)^\top  \\ (\Gamma -f)^\top A_\Gamma  \\ \vdots \\ (\Gamma -f)^\top A_\Gamma ^{n-1}}\,E(t)+\qmx{\bfe^\top_1  \\ \bfe^\top_1  A_K \\ \vdots \\ \bfe^\top_1  A_K^{n-1}}^{-1}\bar\varepsilon(t).
}
}
\end{proof}
%
\section{Adaptive State Observer with LS+DREM Estimator}
\lab{sec4}
%
In this section we complete the design of the proposed adaptive observer combining---in the classical certainty-equivalent way---the results of Proposition \ref{pro1} with the new LS+DREM estimator proposed in Appendix B. More precisely, we apply the LS+DREM estimator of  Proposition \ref{pro3} to the NLPRE \eqref{lre} and replace the estimated parameters in the expression of the state \eqref{xt} to generate the observed state.

The first step is to identify from the NLPRE \eqref{lre} the function  $\bfome(t)$ and the mapping $\calg(\bfthe)$ of \eqref{nlpre}.\footnote{As usual in all adaptive designs \cite{IOASUNbook,SASBODbook}, without loss of generality, we disregard the presence of additive terms that exponentially converge to zero. See \cite[Subsection 4.3.7]{IOASUNbook} for the analysis of its effect.}  This follows trivially defining
\begali{
\nonumber
\bfome(t)& :=\begmat{ \bfome_L^\top (t)& \bfome_N^\top(t)}\in \rea^{n(s+1)}  \\ 
\calg(\bfthe)&:=\begmat{\bfthe \\ \calk(\bfthe)}.
\lab{deflsd}
}
In this case $\bfome(t)\in \rea^{n(s+1)}$ and $\calg:\rea^s \to \rea^{n(s+1)}$. 

It is easy to see that the monotonicity {\bf Assumption A2} is trivially satisfied with the matrix
\begequ
\lab{defbfq}
\bfq=\begmat{I_s& 0_{s \times ns}},
\endequ
which selects only the linear part of the NLPRE \eqref{lre}. 

We are in position to present the main result of the paper. 
	
\begin{proposition}\em
		\lab{pro2}
		Consider the LTV system (\ref{plant}) and the associated NLPRE  \eqref{lre}. Define the certainty equivalent adaptive observer
\begali{
\nonumber
\hat x(t)&= z(t)+ \calw(t,\hat \bfthe(t))\\
&= z(t)+\qmx{\bfe^\top_1  \\ \bfe^\top_1  A_K \\ \vdots \\ \bfe^\top_1  A_K^{n-1}}^{-1}\qmx{(\hat\Gamma(t) -f)^\top  \\ (\hat\Gamma(t) -f)^\top A_{\hat\Gamma}(t) \\ \vdots \\ (\hat\Gamma (t)-f)^\top A_{\hat\Gamma}^{n-1}(t)}\,\left\{L(t)-\left[Q(t)-\begmat{\Omega(t) & P(t)}\right]\qmx{\hat x_{\theta_0}(t)\\ \hat x_{B_0}(t)}\right\},
\lab{hatx}
}
where
\begalis{
{\hat A}_\Gamma(t) &:=\begmat{0_{(n-1)\times 1} & & I_{n-1}\\& \hat \Gamma^\top(t) & },
}
and the estimated parameters 
$$
\hat \bfthe(t)= \begmat{\hat x_{\theta_0}(t) \\\hat x_{B_0}(t)\\ \hat \Gamma(t) }
$$ 
are generated with the LS+DREM interlaced algorithm \eqref{intestt} with the definitions \eqref{deflsd} and \eqref{defbfq}. Assume the vector $\bfome$ satisfies {\bf Assumption  A3} of Appendix B and the system does not exhibit finite escape time. Then, for all initial conditions we have the following.\\

\noindent {\bf (i)} All signals remain \textit{bounded}.\\

\noindent {\bf (i)} The parameter errors satisfy \eqref{parcon}.\\

\noindent {\bf (i)} The state observation error $\tilde x(t):=\hat x(t)-x(t)$ satisfies
$$
\liminf|\tilde x(t)|=0. 
$$
\end{proposition}

 \begin{proof}
As shown in Proposition \ref{pro1}, the state $x(t)$ satisfies  (up to additive exponentially decaying signals) the equation \eqref{xt}, hence \eqref{hatx} is a {\em certainty equivalent} version of it which is obtained replacing the unknown parameters by its current estimate. The proof then is a direct consequence of the {\em exponential} convergence of the parameter estimation errors  \eqref{parcon}, smoothness of all maps  and the assumption that there is no finite escape time. 
\end{proof}
%
\section{Summary of the Adaptive Observer Design}
\lab{sec5}
%
To facilitate its practical implementation in this section we give a summary of the proposed design.\\

\noindent {\bf Data:} $A_\theta(t),A_B(t),h_\theta,h_B$, $u(t)$ and $y(t)$.\\

\noindent {\bf Tuning parameters  for the LRE generator:} $K \in \rea^n,f \in \rea^{n_w}$ such that $A_K$ \eqref{ak} and $A_f$ \eqref{af} are Hurwitz.\\

\noindent {\bf Tuning parameters  for the parameter estimator:} Scalars $\alpha>0$, $f_0>0$ and $\gamma>0$.\\

\noindent {\bf State equations for the LRE generator:} \eqref{con1}, \eqref{con2}, \eqref{con4}.\\ 

\noindent {\bf State equations for the parameter estimator:} \eqref{intestt}.\\

\noindent {\bf Algebraic equations for the LRE generator:} \eqref{con3}.\\

\noindent {\bf Algebraic equations for the parameter estimator:} \eqref{con5}, \eqref{sepnlpre}, \eqref{deflsd}.

%
\section{Numerical example}
\lab{sec6}
%
In this section we will show an example, illustrating the procedure of parameterization of the plant model with time-varying parameters and disturbance generated by an exogeneous system. \\

\noindent {\bf Plant model.} Consider the plant \eqref{plant} with parameters \eqref{ltvsys}, where
$$
A=\qmx{0 & 1 \\
	0 & 0},\qquad
A_\theta=\qmx{-0.001 & 0 \\
	0 & -0.002},\qquad
h_\theta=\qmx{1 & 0 \\ 0 & 1},\qquad
x_\theta(0)=\qmx{-2 \\ -1},
$$
$$
A_B(t)=\qmx{0 & 1 \\
	-1+0.1\sin t & 0},\qquad
h_B=\qmx{1 & 0 \\ 0 & 1},\qquad
x_B(0)=\qmx{0.7 \\ 0.2},
$$
the control signal $u(t)=10+\sin(0.5t)$, 
and the disturbance \eqref{delta} with parameters
$$
S(\rho)=\qmx{0 & 1 \\ \rho & 0},\qquad
h_\delta=\qmx{1 & 0}, \qquad
w(0)=\qmx{-10 \\ 1}, \qquad
\rho =-1.
$$

\noindent  {\bf Parameterization.} We recall that the selection of the vector $\Gamma$ is such that the spectra of $A_\Gamma$ and $S(\rho)$ should coincide. In this examples the characteristic polynomial of $S(\rho)$ is $s^2-\rho$. Therefore, $\Gamma$ can be selected as $\Gamma=\rho\,\bfe_1$ to comply with this requirement. This implies that only one element of $\Gamma$ is unknown and in $\bfthe$ \eqref{bfthe} only one unknown parameter $\rho$ should be included and it takes the form 
$$
\bfthe=\qmx{x_\theta(0) \\ x_B(0) \\ \rho}=\qmx{-2 \\ -1 \\ 0.7 \\ 0.2 \\ -1}.
$$ 
Regarding the nonlinearly parameterized term \eqref{sepnlpre} it takes the form
$$
\Gamma^\top Q(t)\begmat{{x_\theta}_0 \\{x_B}_0}=-\bfe_1^\top Q(t) \qmx{\rho x_\theta(0) \\ \rho x_B(0)},
$$
hence
$$
\bfome_N(t) = -\bfe_1^\top Q(t),\quad {\cal K(\theta)}=\qmx{\rho\,x_\theta(0) \\ \rho\,x_B(0) }=\qmx{2 \\ 1 \\ -0.7 \\ -0.2}.
$$
The description of the regressor equation \eqref{lre} is completed with the definitions given in \eqref{lre}.\\

\noindent   {\bf Parameter estimator and observed state} The estimated parameters $\hat\bfthe$ are generated with the LS+DREM algorithm of Appendix B with the selection matrix $\bfq=\qmx{I_5 & | &0_{5\times 4}}$, which is chosen to retain only the linear part of the regression equation \eqref{lre}. The observation of the states is, finally, given by \eqref{hatx}.\\

\noindent   {\bf Simulation results} The tuning gains used in the simulation are given in the caption of the various figures. In Fig. \ref{fig_1} we show plots of the norm of the estimation error $| \bfthe-\hat\bfthe(t)|$ and the observation error $x(t)-\hat x(t)$. To evaluate the effect of the tuning parameters on the transient performance we repeated the previous simulation significantly changing the gains $f_0$ and $\alpha$. The results are given in  Fig. \ref{fig_2}, which clearly reveal the transient performance degradation---notice the difference in time scales.

Finally, to test the robustness of the adaptive observer to measurement noise we added to the plant output the signal depicted in Fig. \ref{fig_3}. The resulting simulations are given in Fig. \ref{fig_4}, showing very little degradation of the performance. Observing, however, a small steady state error.

 \begin{figure*}[h]
 	\vspace{0mm}
 	\subfloat[][Norm of the parameter estimation error $| \bfthe-\hat\bfthe(t)|$ ]{%
 		\includegraphics[width=0.45\textwidth]{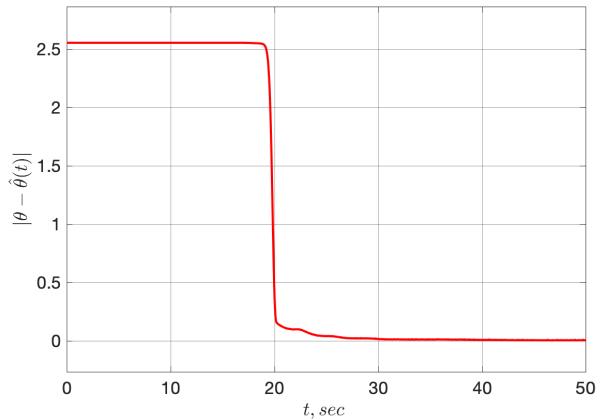}%
 		\label{fig1a}%
 	}
 	\hfill
 	\subfloat[][State observation error $x(t)-\hat x(t)$]{%
 		\includegraphics[width=0.45\textwidth]{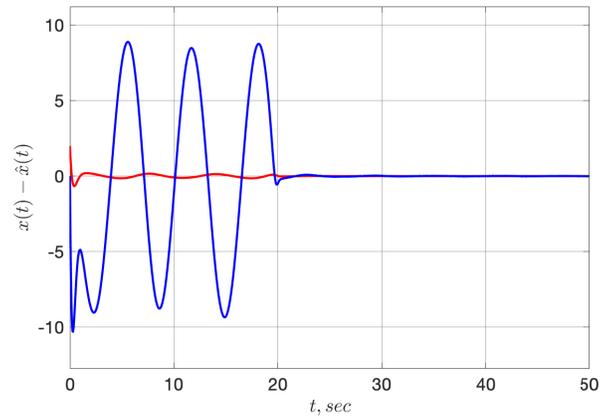}%
 		\label{fig1b}%
 	}
 
 	\caption{The estimation errors for  $K=\col(7.5, 25)$, $f=\col(-1,-2)$, $f_0=0.001$, $\alpha=100$, $\gamma=100$.}
 	\label{fig_1}
 \end{figure*}

\begin{figure*}[h]
	\vspace{0mm}
	\subfloat[][Norm of the parameter estimation error $| \bfthe-\hat\bfthe(t)|$]{%
		\includegraphics[width=0.45\textwidth]{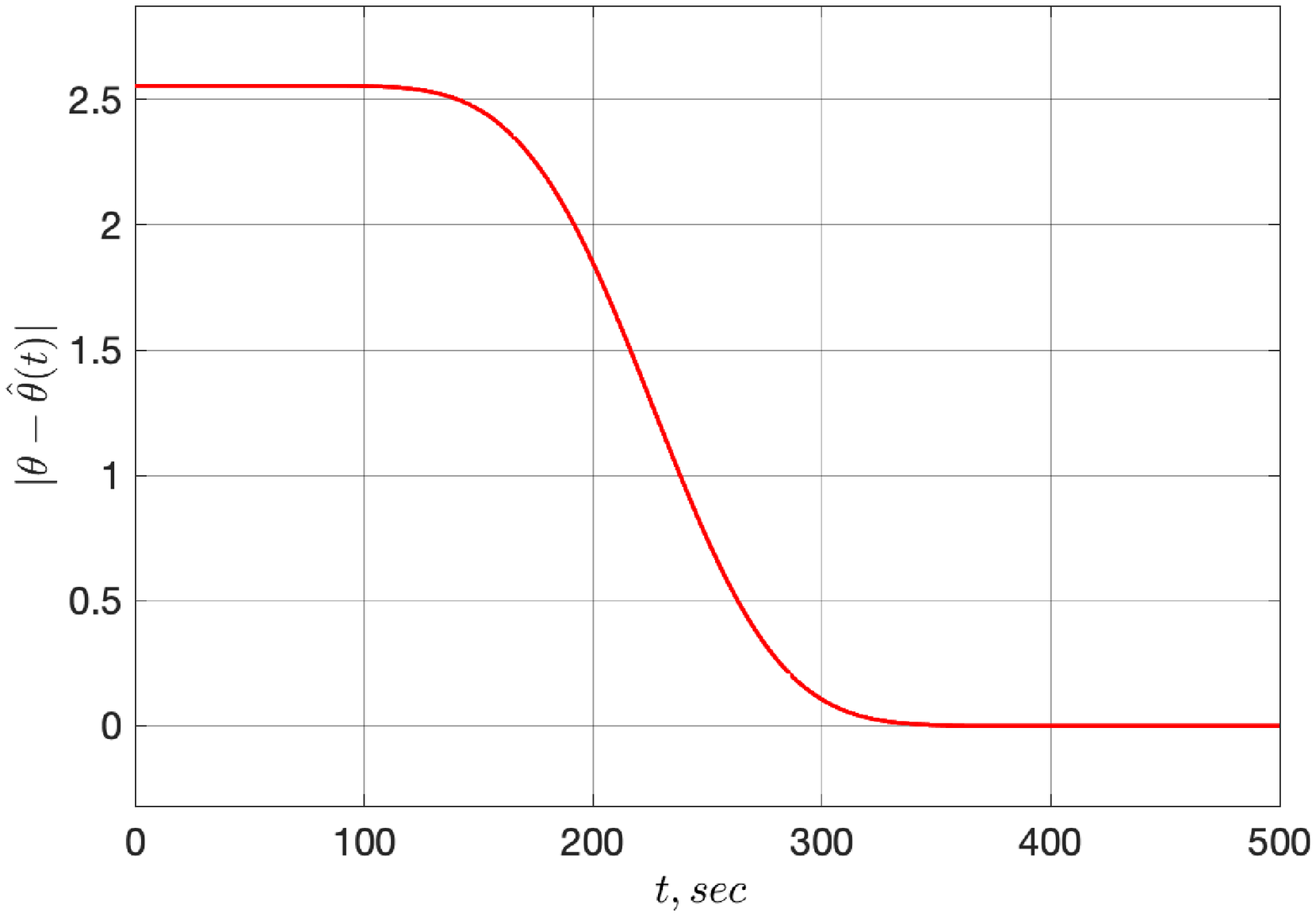}%
		\label{fig2a}%
	}
	\hfill
	\subfloat[][State observation error $x(t)-\hat x(t)$]{%
		\includegraphics[width=0.45\textwidth]{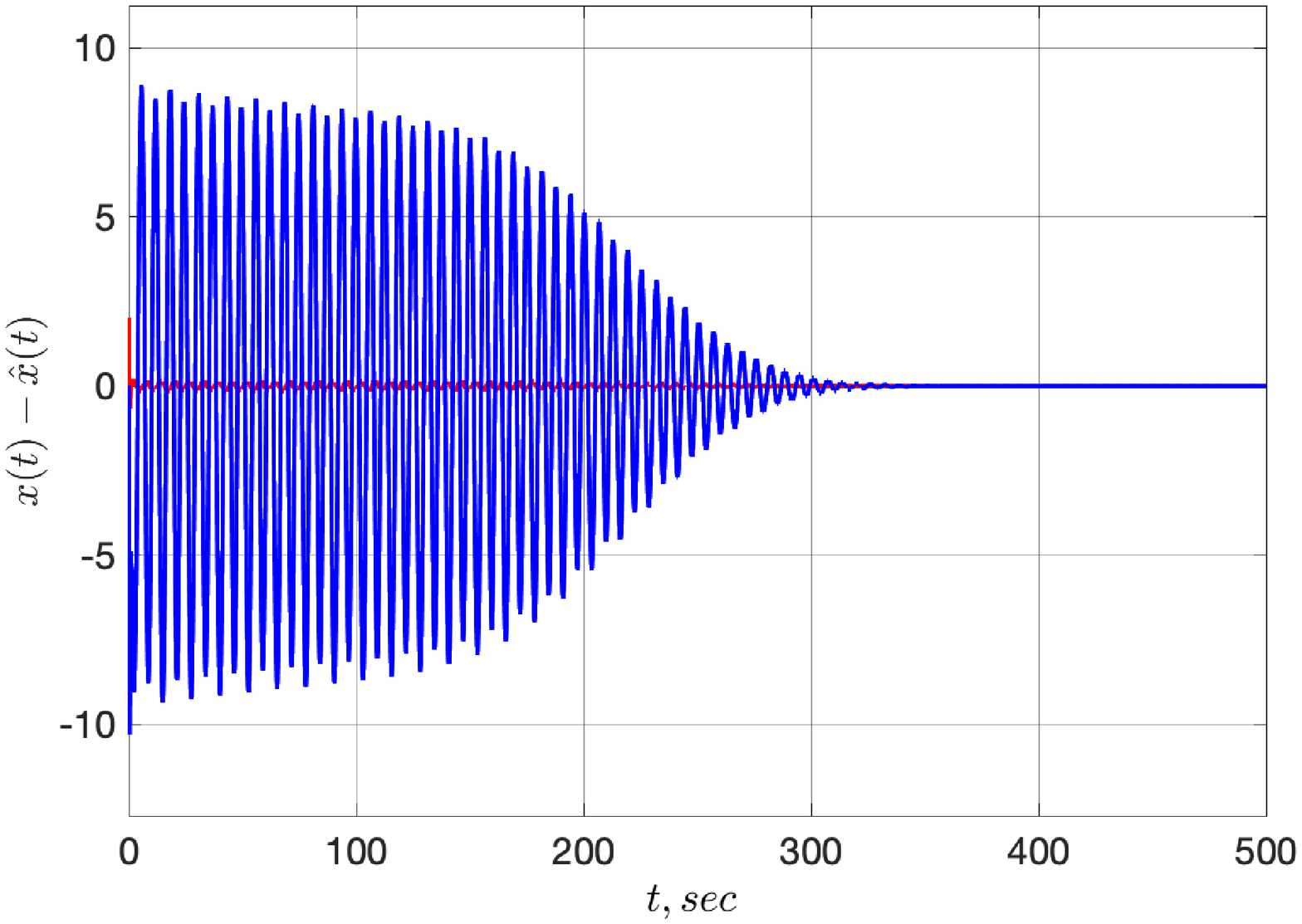}%
		\label{fig2b}%
	}
	
	\caption{The estimation errors for  $K=\col(7.5,25)$, $f=\col(-1,-2)$, $f_0=0.1$, $\alpha=1$, $\gamma=100$.}
	\label{fig_2}
\end{figure*}

  \begin{figure*}[h]
  	\center
 	\vspace{0mm}
 	\includegraphics[width=0.45\textwidth]{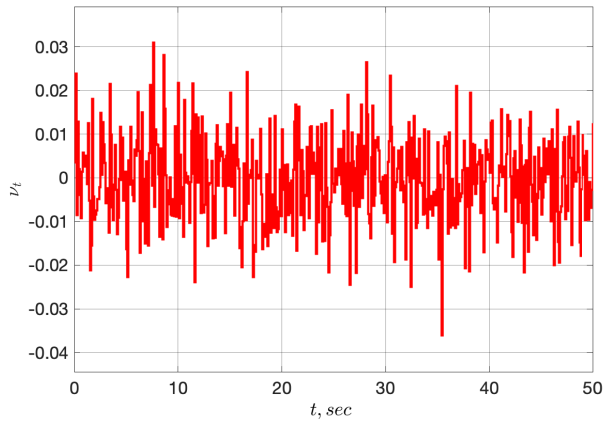}
 	\caption{Plot of the measurement noise}
 	\label{fig_3}
 \end{figure*}
 
  \begin{figure*}[h]
 	\vspace{0mm}
 	\subfloat[][Norm of the parameter estimation error $| \bfthe-\hat\bfthe(t)|$ ]{%
 		\includegraphics[width=0.45\textwidth]{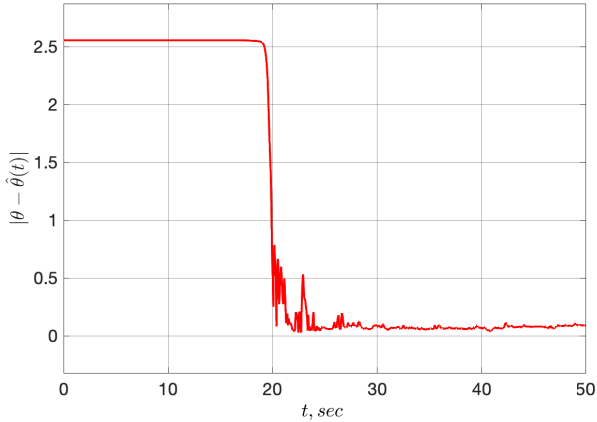}%
 		\label{fig4a}%
 	}
 	\hfill
 	\subfloat[][State observation error $x(t)-\hat x(t)$]{%
 		\includegraphics[width=0.45\textwidth]{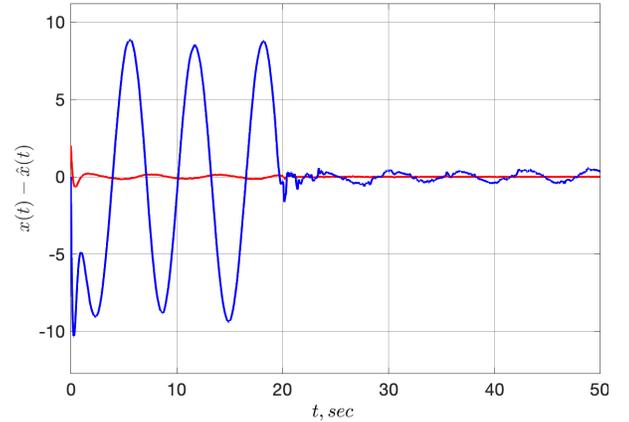}%
 		\label{fig4b}%
 	}
 	
 	\caption{The estimation errors with noisy measurement for  $K=\col(7.5,25)$, $f=\col(-1,-2)$, $f_0=0.001$, $\alpha=100$, $\gamma=100$.}
 	\label{fig_4}
 \end{figure*}
 %
 \section{Concluding Remarks}
 \lab{sec7}
 %
We have presented in the article a new general technique of state observation of LTV systems disturbed by an additive signal generated by an LTI exosystem. The system matrices and the matrix of the exosystem are assume to be unknown. The main analytical tool employed for the solution of the problem is the GPEBO technique that translates the state observation problem in a parameter estimation task, for which a regression equation---in the present case nonlinearly parameterized---is developed. To overcome the difficulty of estimating the parameters in this case we propose to use the DREM technique. In particular, we propose a new least-squares-based DREM-like estimator that allows us to solve this task.
 
Thanks to this [LS+DREM] scheme the estimation of the unknown constant parameters is achieved very rapidly and robustly with respect to measuring noise. It is also shown that the speed of convergence  may be easily regulated suitably selecting the tuning coefficients of the estimator---which consists only of four scalar parameters, whose role is clearly identified. Some simulation results reveal the efficiency of the proposed approach.

In future works we are planning to investigate the extension to the nonlinear setting of our work that is mentioned in Remark \ref{rem_nl}.
 
%

\appendix

\section{Auxiliary Lemmata}
%
In this Appendix we give two preliminary lemmata that are used in the construction of the LRE of Section \ref{sec3}. Although the first one of these  may be found in the literature \cite{ISIbook},  for the sake of completeness we also give its proof.

\begin{lemma} \em
Consider a stable LTI system
\[
\dot x(t) = Ax(t) + Bu(t)
\]
with $x(t)\in \rea^n$, driven by the exosystem \eqref{exosys}
\begali{
\nonumber
\dot w(t) &= Sw(t)\\
u(t) &= Qw(t),
\lab{exosys}
}
with $w(t)\in \rea^n$, where the spectra of $A$ and $S$ are disjoint. Then, $x(t)$ can be split as
\[
x(t) = \Pi w(t) + \tilde x(t)
\]
where $\Pi$ is the unique solution of the Sylvester equation
\[
\Pi S = A\Pi + BQ
\]
and  $\tilde x(t)$ is exponentially decaying.
 \end{lemma}

 \begin{proof}
Form the composition
\[
\ba{rcl}
\dot w(t) &=& Sw(t)\\
\dot x(t) &=& Ax(t) + BQw(t)\ea
\]
The change of variables
\[
 \tilde x(t) = x(t)-\Pi w(t)
\]
yields the diagonal system
\[
\ba{rcl}
\dot w(t) &=& Sw(t)\\
\dot  {\tilde x}(t) &=& A  \tilde x(t).
\ea
\]
The proof is completed recalling that $A$ is Hurwitz.
 \end{proof}

The construction introduced in the proof of Lemma 1 can also be used in the proof of the following Lemma, from which the expression (\ref{Psi2}) can be deduced.
\begin{lemma}\em
	 Consider a stable LTI system with the state $x(t)\in \rea^n$
	\[
	\dot x(t) = A_f x(t) + \bfe_n u(t),
	\]
	driven by the exosystem
	\[\ba{rcl}
	u(t) &= &Qw(t)+Me(t), \\
	\dot w(t) &=& Sw(t),\\
	\dot e(t) &=& A_Ke(t)
	\ea
	\]
	where $w(t) \in \rea^n$, the spectra of $A_f$ and $S$ are disjoint,  $A_f$ is a Hurwitz matrix having the form \eqref{af} and the matrix $A_K$ is an arbitrary Hurwitz matrix. 
	Then,
	\begite
		\item[(i)] $x(t)$ can be split as
	\[
	x(t) = {\Pi_f} w(t) + \tilde x(t)
	\]
	where ${ \Pi_f}$ is the unique solution of the Sylvester equation
	\[
	{ \Pi_f} S = A_f{ \Pi_f} + \bfe_nQ
	\]
and $(\tilde x,e)$ satisfies
\begin{equation} \label{triangular}
	\ba{rcl}
	\qmx{\dot  {\tilde x}(t) \\ \dot e(t)} &=& \qmx{A_f  & \bfe_n M \\ 0 & A_K} \qmx{\tilde x(t) \\ e(t)}.
	\ea
\end{equation}
	\item[(ii)] $x(t)$ satisfies
	\[
	\dot x(t) = A_\Gamma  x(t) + \bfe_n \varepsilon(t),
	\]
	where $A_\Gamma$ is given in \eqref{agam}, the spectra of $A_\Gamma $ and $S$ coincide, and  $\varepsilon(t)$ is exponentially decaying.	
\endite	
\end{lemma}

 \begin{proof}
	With Lemma 1 in mind, consider the composition
	\[
	\ba{rcl}
	\dot w(t) &=& Sw(t)\\
	\dot e(t)&=&A_Ke(t)\\
	\dot x(t) &=& A_fx(t) + \bfe_nQw(t)+\bfe_n Me(t).\ea
	\]
	The change of variables
	\[
	\tilde x(t) = x(t)-\Pi w(t)
	\]
	proves (\ref{triangular}).

In order to prove claim (ii), observe that
	\begalis{
		x_1(t) =\bfe^\top_1x(t) = h^\top w(t) +  \bfe^\top_1 \tilde x (t)
	}
	where $h^\top= \bfe^\top_1  \Pi_f$.	Moreover, in view of the special form of $A_f$, we have $x_{i+1}(t)=\dot x_i(t)$ for $i=1,\ldots, n-1$.

As a consequence
	\begalis{
		x_1&=h^\top w(t) +  \bfe^\top_1  \tilde x(t),\\
		x_2&=h^\top S w(t) +  \bfe^\top_1  [\tilde x(t)]^{(1)}\\
		x_3&=h^\top S^2 w(t) +  \bfe^\top_1  [\tilde x(t)]^{(2)}\\
		&\quad\vdots\nonumber \\
		x_n&=h^\top S^{n-1} w(t) +  \bfe^\top_1  [\tilde x(t)]^{(n-1)}
	}
	and, finally,
	\beal{
		\lab{chielldot}
		\dot x_n=h^\top S^n w(t) +  \bfe^\top_1  [\tilde x(t)]^{(n)}\,.
	}
Consider now the charactersistic polynomial of matrix $S$:
	\[
	\gamma(\lambda)=\det(\lambda I-S)=\lambda^n+\gamma_1\lambda^{n-1}+\gamma_2\lambda^{n-2}+\cdots +\gamma_n\lambda^0.
	\]
and recall that, thanks to the Cayley--Hamilton theorem,
	\[
	S^n=-\gamma_1S^{n-1}-\gamma_2S^{n-2}-\cdots -\gamma_n S^0.
	\]
	
	Substituting the latter into \eqref{chielldot} yields
	\begalis{
		\dot x_n&=h^\top (-\gamma_1S^{n-1}-\gamma_2S^{n-2}-\cdots -\gamma_n S^0) w(t) +  \bfe^\top_1  [\tilde x(t)]^{(n)}\nonumber\\
		&=-\gamma_nx_1-\gamma_{n-1}x_2-\cdots-\gamma_1x_n+ \bfe^\top_1  \Bigl[[\tilde x(t)]^{(n)}+\sum_{i=1}^n \gamma_i[\tilde x(t)]^{(n-i)}\Bigr]\nonumber\\
		&=-\gamma_nx_1-\gamma_{n-1}x_2-\cdots-\gamma_1x_n+\varepsilon(t),
	}
	where
	\beal{
	\lab{vareps}
		\varepsilon(t)&=\bfe^\top_1  \Bigl[[\tilde x(t)]^{(n)}+\sum_{i=1}^n \gamma_i[\tilde x(t)]^{(n-i)}\Bigr].
	}
	
In this way, we have shown that
	\begalis{
		\dot x(t)&=A_\Gamma x(t)+\bfe_n \varepsilon(t),
	}
	in which $\Gamma =\qmx{-\gamma_n & -\gamma_{n-1} & ... & -\gamma_1}$.

To complete the proof of (ii) it remains to check that $\varepsilon(t)$ is exponentially decaying. But this is a straightforward consequence of (\ref{triangular}), which we rewrite as
\begequ
\lab{dotxi}
\dot \xi(t)= F_{\rm c}\xi(t)
\endequ
with 
\[
\xi(t) := \begmat{\tilde x(t) \\e(t)},\;F_{\rm c}:=\qmx{A_f  & \bfe_n M \\ 0 & A_K}, 
\]
and $F_{\rm c}$ a Hurwitz matrix. Now, $\tilde x(t)$ can expressed as $\tilde x(t) = H_{\rm c}\xi(t)$ with  $H_{\rm c}:= \begmat{I_n &0_{n \times n}}$.
Thus
 \[
[\tilde x(t)]^{(i)} =  H_{\rm c}F_{\rm c}^i\xi(t).
\]
Replacing the identity above in \eqref{vareps} yields
\begequ
\lab{vareps1}
		\varepsilon(t)= h^\top_\varepsilon \xi(t),
\endequ
where we defined the vector
\begalis{
	h_\varepsilon:= \bfe^\top_1H_{\rm c} \Bigl[F_{\rm c}^n+\sum_{i=1}^{n} \gamma_iF_{\rm c}^{n-i}\Bigr].
	}
Since $\xi(t)$ is exponentially decaying, claim (ii) is proven.
\end{proof}

\begrem
\label{rem_nl}
	It is worth stressing that, taking advantage of an important result of \cite{MPI},  the property indicated in Lemma 2 can be extended to the case in which a stable $d$-dimensional system 
	$$\dot x(t) = A_f x(t) + \bfe_d u(t)$$ 
	is driven by a nonlinear exosystem of the form
	
	\[\begin{array}{rcl}
		u(t) &= &q(w(t))+Me(t), \\
		\dot w(t) &=& s(w(t)),\\
		\dot e(t) &=& A_Ke(t)
	\end{array}
	\]
	
	Also in this case, in fact, it is possible to show that - if $d$ and $A_f$ are appropriately chosen - $x(t)$ can be seen as the solution of a  system $\dot  x = \hat s(x) + \bfe_d\varepsilon$ in which  $\hat s(\cdot)$ is a suitable ``copy" of $s(\cdot)$ and $\varepsilon$ is asymptotically decaying to 0. 
\endrem
%
\section{A Least-Squares-Based DREM Parameter Estimator}
%
In this appendix we propose a new parameter estimator for separable, NLPRE which is a combination of a least-squares search and the mixing procedure of DREM. The least-squares search is used to construct the extended regressor---this in contrast with previous versions of DREM where this regressor is obtained applying a bank of linear operators to the original regressor equation, see   \cite{ORTetaltac} for a recent overview of the DREM procedure. There are two significant advantages of the new estimator with respect to the previous DREM algorithms. First, we avoid the step of selecting the bank of linear operators, which significantly affect the performance of the estimator and whose choice is far from clear. Second, it is widely recognized that the transient performance of the least-squares search significantly outperforms the gradient-based one.    

The following standing assumptions are imposed throughout the appendix.\\  	

\noindent {\bf Assumption A1} [A NLPRE] The existence of a separable, NLPRE
\begequ
\lab{nlpre}
\bfy(t)=\bfome(t) \calg(\bfthe),
\endequ
where the signals $\bfy(t) \in \rea^m$ and $\bfome(t) \in \rea^{m \times p}$ are measurable, the mapping $\calg:\rea^q \to \rea^p,\;q \leq p$ is known but the constant parameters $\bfthe \in \rea^q$ are unknown.\\
	
\noindent {\bf Assumption A2} [Monotonicity] There exists a matrix  $\bfq\in\mathbb{R}^{q\times p}$ such that mapping $\calg(\bfthe)$ verifies the linear matrix inequality
\begali{
	\lab{monpro-S}
	\bfq \nabla\calg(\bfthe) + \nabla^\top  \calg(\bfthe)\bfq^\top  \geq \rho I_q >  0,\;\forall \; \bfthe \in \rea^q,
}
for some $\rho \in \rea_{>0}$. Consequently  \cite{DEM,PAVetal}, The mapping $\bfq\calg(\theta)$ is {\em strongly monotone}, that is, 
\begali{
\lab{monpro}
(a-b)^\top  &\left[\bfq \calg(a) -\bfq \calg(b)\right] \geq \rho|a-b|^2 >  0,\;\forall \; a,b \in \rea^q,\;a \neq b.
}

\noindent {\bf Assumption A3} [Interval Excitation] 	\cite[Definition 3.1]{TAObook}	The regressor matrix $\Omega$ of the NLPRE \eqref{lre} is  interval exciting (IE). That is, there exists constants $C_c>0$ and $t_c\geq t_0>0$ such that
	\begalis{
		&\int_{t_0}^{t_0+t_c} \bfome^\top(s) \bfome(s)  ds \ge C_c I_p.
	}
	
	\begin{proposition}\em
		\lab{pro3}
		Consider the NLPRE (\ref{nlpre}) with  $\calg(\bfthe)$ satisfying  {\bf Assumption A2} and $\bfome$ verifying {\bf Assumption  A3}. Define the LS+DREM interlaced estimator
		\begsubequ
		\lab{intestt}
		\begali{
			\lab{thegt}
			\dot{\hat \bfthe}_g(t) & ={\alpha} F(t) \bfome^\top(t)  (\bfy(t)-\bfome(t) \hat\bfthe_g(t)),\; \hat\bfthe_g(0)=:\bfthe_{g0} \in \rea^p\\
			\lab{phit}
			\dot {F}(t)& =  -{\alpha} F(t) \Omega^\top(t)  \Omega(t)  F(t),\; F(0)={1 \over f_0} I_p \\
			\lab{thet}
			\dot{\hat \bfthe}(t) & =\gamma \bfq \Delta(t) [\caly(t) -\Delta(t) \calg(\hat\bfthe(t)) ],\; \hat\bfthe(0)=:\bfthe_0 \in \rea^q,
		}
		\endsubequ
		with tuning gains $f_0>0$ , ${\alpha>0}$ and $\gamma >0$,  and we defined
		\begsubequ
		\lab{aydelt}
		\begali{
			\lab{delt}
			\Delta(t) & :=\det\{I_p-f_0F(t)\}\\
			\lab{yt}
			\caly(t) & := \adj\{I_p- f_0F(t)\} (\hat\bfthe_g(t) -  f_0F(t) \bfthe_{g0}),
		}
		\endsubequ
		where $ \adj\{\cdot\}$ denotes the adjugate matrix. Define the parameter error $\tilthe(t):=\hat \bfthe(t)-\bfthe$. Then, for all initial conditions $\bfthe_{g0} \in \rea^p$ and $\bfthe_{0} \in \rea^q$, we have that
		\begequ
		\lab{parcon}
		\lim_{t \to \infty}\tilthe(t)=0,\;(exp),
		\endequ
		with all signals \textit{bounded}.
	\end{proposition}

 \begin{proof}
With some abuse of notation, define the signal
$$
\tilde \calg(t):=  \hat\bfthe_g(t) - \calg(\bfthe),
$$
whose derivative is given by
\begalis{
\dot{\tilde \calg}(t)&=-\alpha F(t)\bfome^\top (t) \bfome (t)\tilde\calg(t),
}
where we used  \eqref{nlpre} and \eqref{thegt}. Now, from the fact that \cite[Theorem 4.3.4]{IOASUNbook},
$$
{d \over dt}(F^{-1}(t)\tilde \calg(t))  =0
$$
we have 
$$
\tilde \calg(t)  =f_0F(t)\tilde\calg(0),
$$
which may be rewritten as the extended NLPRE
\begali{
\lab{keyide}
(I_p-f_0F(t))\calg(\bfthe) &=\hat \bfthe_g(t) - f_0F (t)\bfthe_{g0}\,,
}
Following the DREM procedure we multiply \eqref{keyide} by   $\adj\{I_p-f_0F(t)\}$ to get the following  NLPRE
\begequ
\lab{ydel}
\caly(t)= \Delta(t)\calg(\bfthe),
\endequ
where we used \eqref{delt} and \eqref{yt}. Replacing \eqref{ydel} in \eqref{thet} we get
\begalis{
\dot{\hat \bfthe} (t)& =-{ \gamma \Delta^2(t)} \bfq[ \calg(\hat\bfthe(t)) - \calg(\bfthe)] .
}
To analyze the stability of the latter system define the Lyapunov function candidate 
$$
U(\tilde \bfthe) := \frac{1}{2\gamma} |\tilde \bfthe|^2.
$$
Computing its time derivative yields
\begalis{
\dot U (t) & =  - \Delta^2[ \hat \bfthe - \bfthe]^\top \bfq {[ \calg(\hat \bfthe(t) ) - \calg(\bfthe)]} \\
& \leq  -  \Delta^2(t) \rho | \tilde \bfthe(t)|^2  \\
& =  - 2\rho \gamma{\Delta}^2(t) U(t),
}
where we invoked the {\bf Assumption A.2} of strong monotonicity property \eqref{monpro} of $\bfq \calg(\bfthe)$ to get the first bound. To complete the proof, we invoke the Comparison Lemma \cite[Lemma 3.4]{KHAbook} that yields the bound
$$
U(t+t_c) \leq e^{-2 \rho\gamma \int_t^{t+t_c} \Delta^2(s)ds}U(t),
$$
which ensures $\lim_{t \to \infty}\tilthe(t)=0\;(exp)$ if  $\Delta(t)$ is PE. The latter condition follows from the assumption that $\bfome(t)$ is IE and \cite[Lemma 3.5]{TAObook}, which ensures the following
\begequ
\lab{taolem}
\bfome(t) \in IE\;\Rightarrow\;\Delta(t) >0,\;\forall t \geq t_0+t_c.
\endequ
Consequently, $\Delta(t)$ is PE.
\end{proof}

\begrem
	
Proposed estimator \eqref{phit} may be compared with the adaptive estimation algorithm presented in \cite{PYRHYP}, where the problem of frequency estimation of a periodical signals corrupted by a noise was considered. Indeed, for the case $m=p=1$ with $F\in\rea$, $\Omega\in\rea$ the algorithm \eqref{phit} becomes
$$
\dot F = -\alpha \Omega^2 F^2.
$$

In \cite{PYRHYP} a similar tuning rule was deduced (see equation (45)). But in \cite{PYRHYP} was derived with other (heuristic) arguments, from an idea of iterative decreasing an adaptation gain at special time instants \cite{EFITAC}. Contineous extention of the iterative approach \cite{EFITAC} yielded  \cite{PYRHYP} which is a very particular scenario of the general approach \eqref{phit} proposed in this paper.
\endrem

\end{document}